\documentclass{amsart}
\usepackage[margin=.8in]{geometry}
\usepackage{graphicx}
\usepackage{fancyhdr, multicol}
\usepackage{float}

\author{Kevin Shu}
\title{Causal Channels}

\usepackage{amsmath, amsthm, amsfonts, amssymb}
\usepackage{tikz}
\usetikzlibrary{arrows.meta}
\usepackage{mdframed}
\usepackage[parfill]{parskip}

\newtheorem{theorem}{Theorem}
\newtheorem*{reftheorem}{Theorem}

\newtheorem{lemma}{Lemma}

\newcommand{\R}{\mathbb{R}}

\newcommand{\Joint}{\mathcal{J}}
\newcommand{\Cond}{\mathcal{C}}
\newcommand{\Mu}{\Joint(\Omega, \Gamma)}
\newcommand{\Nu}{\Cond(\Omega\;|\;\Gamma \times \Omega)}
\newcommand{\Obs}{\Joint(\Omega, \Omega)}
\newcommand{\Int}{\Cond(\Omega\;|\; \Omega)}
\newcommand{\intmap}{\phi_{\text{int}}}

\DeclareMathOperator*{\argmin}{arg\,min}
\DeclareMathOperator*{\sgn}{sign}
\DeclareMathOperator*{\pr}{Pr}
\DeclareMathOperator*{\Do}{do}

\DeclareMathOperator*{\supp}{supp}
\DeclareMathOperator*{\conv}{conv}

\DeclareMathOperator*{\val}{val}

\begin{document}
\begin{abstract}
    We consider causal models with two observed variables and one latent variables, each variable being discrete, with the goal of characterizing the possible distributions on outcomes that can result from controlling one of the observed variables.
    We optimize linear functions over the space of all possible interventional distributions, which allows us find properties of the interventional distribution even when we cannot uniquely identify what it is.
    We show that, under certain mild assumptions about the correlation between controlled variable and the latent variable, the resulting interventional distribution must be close to the observed conditional distribution in a quantitative sense.
    Specifically, we show that if the observed variables are sufficiently highly correlated, and the latent variable can only take on a small number of distinct values, then the variables will remain causally related after passing to the interventional distribution.
    Another result, possibly of more general interest, is a bound on the distance between the interventional distribution and the observed conditional distribution in terms of the mutual information between the controlled variable and the latent variable, which shows that the controlled variable and the latent variable must be tightly correlated for the interventional distribution to differ significantly from the observed distribution.
    We believe that this type of result may make it possible to rigorously consider `weak' experiments, where the causal variable is not entirely independent from the environment, but only approximately so.
    More generally, we suggest a connection between the theory of causality to polynomial optimization, which give useful bounds on the space of interventional distributions.
\end{abstract}
\maketitle
\section{Introduction}
The goal of this paper will be to analyze a simple model of causality using tools from mathematical optimization and see what assumptions are needed to understand causal effects from observational studies.

%
%

\subsection{Model and Overview of Results}
We consider 3 random variables ($U$, $X$, $Z$), which are discrete random variables on finite sets.
To completely specify the behavior of these variables, we would need to specify how all 3 of the variables depend probabilistically on each other, but we will only be given the joint distribution of $X$ and $Z$ is given, and the dependencies between $U$ and the other variables are unknown.

Given the joint distribution of $X$ and $Z$, there are many possible models that are consistent with this joint distribution, and the set of all of these consistent models can be defined in terms of quadratic equations and linear inequalities corresponding to the independence assumptions in the model (see section \ref{sec:setup} for details). 
The \textbf{interventional distribution}, denoted $\pr(Z | \Do(X))$ can also be defined in terms of quadratic polynomials.

Though it is impossible to uniquely determine what the causal effect of $X$ will be on $Z$ in this model, we can nevertheless try to understand the space of possible interventional distributions.
We can show that it is possible to describe the space of interventional distributions in terms of a finite collection of higher degree polynomial inequalities, but enumerating these inequalities is intractable in general.
We will mostly be concerned with linear inequalities which are valid on the space of interventional distributions.
Such linear inequalities define the convex hull of the space of interventional distributions.
We will consider three settings in which it is possible to understand this convex hull in greater detail.

Firstly, in a toy model where the two variables are `perfectly correlated', we can completely determine all possible interventional distributions that result from this observational distribution, as they will turn out to be union of a exponentially many convex polytopes. 
Secondly, we consider a situation in which $X$ and $Z$ are close to perfectly correlated, but the number of possible values the confounding variables $U$ is small, and we show that it is not possible for the causal distribution of $Z$ to be independent of $X$.
Finally, we consider a situation in which $X$ and $U$ are close to independent, in the sense that their mutual information is small, then we obtain that all interventional distributions compatible with the observed distribution are close (in an $\mathcal{L}_1$ sense) to the observed conditional distribution.

\subsection{Layout}
The layout of this paper is as follows: in section \ref{sec:setup}, we will precisely define causal models and the particular class of models we study.
In section \ref{sec:results}, we will state our main results. In section \ref{sec:geometry}, we will use geometric language to state the problems involved precisely, and describe the connection between our problems and linear programming. The remaining sections are devoted to proofs of the stated results.


\section{Background and Setup}
\label{sec:setup}

\subsection{Background}
In 2000, Judea Pearl initiated the study of causality using graphical models.\cite{pearl2009causality}
In causality theory, one considers a collection of variables, which are divided into those which can be observed, and those which are latent, and wants to know how the observed variables will interact with each other when changed by some external force.
In graphical models of causality, one assumes that the conditional distribution of one variable given all of the other variables is fixed in advance, and that there is some directed acyclic graph $G$ on the random variables, so that each variable $v$ is conditionally independent of all of the other variables given the variables in the neighborhood of $v$.
The input to a causal model is a joint distribution of the observed variables which we will call the \textbf{observed distribution}.

Given a graphical model, and two fixed variables, $X$ and $Z$, it is possible to determine whether observed distribution \textbf{uniquely determines} the interventional distribution of $X$ and $Z$.
If it does, the model is said to be \textbf{identifiable}.
When there is a unique interventional distribution, it can be found algorithmically using the methods in \cite{shpitser2006identification} and \cite{huang2012pearl}.
These identifiability results require some stringent conditions on the relationships between the observed and latent variables, which may be difficult to confirm in practice.
It was shown in \cite{schulman2016stability} that even if if graphical model is identifiable, the interventional distribution might not be stable under perturbing the observed data.
The graphical model we consider here will not be identifiable, but we will be able to understand the causal connection between the two variables without performing an experiment.

Algebraic approaches to causality have been used in the past; an early example can be found in \cite{garcia2005algebraic}, which describes an approach using algebraic geometry to model the equations defining the interventional distribution space. 
\cite{LeeSpekkens} also use algebraic methods to understand the functional dependencies between two observed variables, though they work in a different setting.
As far as we know, the usage of polynomial optimization theory in causality is novel.

\subsection{Graphical Models}
We will consider a graphical model with just 3 variables, where we think of two of them as observed, and one as latent.
We can represent this situation diagrammatically as follows:

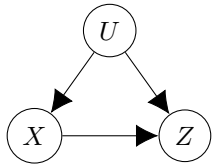
\begin{figure}[H]
    \centering
    \begin{tikzpicture}
        \tikzstyle{every node}=[draw,shape=circle];
        \node (X) at (0, 0) {$X$};
        \node (U) at (1, 1.4) {$U$};
        \node (Z) at (2, 0) {$Z$};

        \draw[-{Latex[length=3mm,width=3mm]}] (U) -- (X);
        \draw[-{Latex[length=3mm,width=3mm]}] (U) -- (Z);
        \draw[-{Latex[length=3mm,width=3mm]}] (X) -- (Z);
    \end{tikzpicture}
    \caption{Graphical Model under consideration.}
    \label{fig:name}
\end{figure}
There are three major variables which are consider in this model:
\begin{enumerate}
    \item $X$ is the variable we want to control.
    \item $Z$ is the effect variable. We want to understand how controlling $X$ will affect the distribution over $Z$.
    \item $U$ is the latent variable, or confounding variable.
\end{enumerate}

This directed acyclic graph is meant to indicate a dependency order for the variables: $X$ depends causally on $U$ and $Z$ depends causally on $X$ and $U$.
Throughout this discussion, we will assume that $X$ and $Z$ both take values in some finite set $\Omega$ with $n$ elements, and $U$ takes values in some finite set $\Gamma$ with $k$ elements.

In this setting, if $U$ can be arbitrarily correlated with $X$ and $Z$, then any observed effect of $X$ on $Z$ might vanish when $X$ is controlled, so we will examine some criteria for which we can still make some conclusion.

\subsection{Causal Models}
To specify a \textbf{causal model} over this graph, we will need to describe how the various variables depend on one another probabilistically.

For a given finite set $S$, we will use $\Delta(S)$ to denote the probability simplex parametrizing probability distributions on elements of $S$, i.e.
\[
    \Delta(S) = \{(a_i) \in \R_+^S : \sum_i a_i = 1\}
\]

Given two sets, $S_1, S_2$, a joint probability distribution on $S_1$ and $S_2$ is just a probability distribution on the product set $S_1 \times S_2$.
Hence, we define the space of joint distributions on $S_1$ and $S_2$ to be $\Joint(S_1, S_2) = \Delta(S_1 \times S_2)$.

Similarly, if we have sets $S_1$ and $S_2$, then we can consider conditional distributions of $S_1$ on $S_2$ as giving, for each element of $S_2$, a probability distribution on $S_1$.
Hence, we also define the space of conditional distributions as $\Cond(S_1|S_2) = \Delta(S_1)^{S_2}$.

We can specify a causal model by a pair $(\mu, \nu)$ where
\begin{enumerate}
    \item $\mu \in \Mu$, which we think of as being the joint distribution of $X$ and $U$. We will use both the notation $\mu_{x, u}$ and $\mu(X=x, U=u)$ to denote the joint probability that $X = x$ and $U = u$.
    \item $\nu \in \Nu$, which we think of as being the conditional distribution of $Z$ conditioned on each possible value of $X$ and $U$. We will use $\nu_{z,x,u}$ and $\nu(Z=z | X=x, U=u)$ to denote the conditional probability that $Z = z$, given that $X = x$ and $U = u$.
\end{enumerate}

\textbf{Remark:} In traditional probabilistic notation, we would use $Pr(X, U)$ to denote $\mu$ and $Pr(Z | X, U)$ to denote $\nu$.
The advantage of the traditional notation is that we can use notation such as $Pr(X | U)$ to simplify the notation for certain arithmetic operations.
We adopt this notation to make it clearer that these parameters are unknown and may change in different places.
We will use the probabilistic-style notation to state the results, and then switch to the vector-style notation in the proofs to save space.
 
A pair of model parameters, $(\mu, \nu)$ determines two important auxiliary quantities: the observational distribution $\pi \in \Joint(\Omega, \Omega)$, and the interventional distribution, $\zeta \in \Int$.

The \textbf{observational distribution} is defined as
\[
    \pi(X = x, Z=z) = \sum_{u \in \Gamma} \mu(X=x, U=u)\nu(Z=z| X=x, U=u).
\]

Given a distribution $\pi \in \Obs$, we say that the model defined by $(\mu, \nu)$ is compatible with $\pi$ if $\pi$ is in fact the observational distribution obtained by considering $\mu$ and $\nu$. To fix some notation, let 
\[
    M_{\pi} = \{(\mu, \nu) \in \Mu \times \Nu : (\mu, \nu)\text{ is compatible with }\pi\}
\]

The \textbf{interventional distribution} is defined as 
\[
    \zeta(Z=z | \Do(X=x)) = \sum_{u \in \Gamma} \mu(U=u)\nu(Z=z|X=x, U=u)
\]
A distribution $\zeta \in \Int$ is said to be compatible with a distribution $\pi \in \Obs$  if there exists some $\mu, \nu \in M_{\pi}$ so that $\zeta$ is the interventional distribution of $(\mu, \nu)$ and so that $\pi$ is the observed distribution of $(\mu, \nu)$.

Let
\[
    I_{\pi} = \{\zeta \in \Int : \zeta\text{ is compatible with }\pi\}
\]
$I_{\pi}$ denote the set all interventional distributions that are associated to a given observed distributing $\pi$.

We will give more detailed descriptions of $M_{\pi}$ and $I_{\pi}$ in section \ref{sec:geometry}.
\section{Results in Detail}
\label{sec:results}

\subsection{Perfect Channels}
Say that $X$ and $Z$ are \textbf{perfectly correlated} if for each $x, z \in \Omega$ with $x \neq z$,
\[
    \pi(X=x, Z=z) = 0.
\]
That is, in the observational distribution, $X$ and $Z$ are observed to always be equal.

In \textbf{theorem \ref{thm:perfect_full}}, we express $I_{\pi}$ as the union of $(2^{k}-1)^n$ polytopes, each of whose extreme rays can be fully specified.
In this sense, it is possible to fully characterize the space of interventional distributions in this case.
As a note, the content of this section is not needed to understand the remainder of the results.

The convex hull of $M_{\pi}$ can be described in a more compact form.
For each function $f : \Omega \rightarrow \Gamma$, define the polytope
\[
    Q_f = \{\zeta \in \Int : \forall x \in \Omega, \zeta(Z=x|\Do(X=x)) \ge \sum_{x' : f(x') = f(x)}\pi(X=x')\}
\]

\begin{reftheorem}(Theorem \ref{thm:perfect_conv})
    If $X$ and $Z$ are perfectly correlated under $\pi$, then $\bigcup_{f \in \Gamma^{\Omega}} Q_f \subseteq I_{\pi}$ and the convex hull of $M_{\pi}$ is equal to the convex hull of $\bigcup_{f \in \Gamma^{\Omega}} Q_f$.
\end{reftheorem}

Note that in particular, it is possible for $\zeta(Z=x|\Do(X=x)$ to be as small as $\frac{1}{k}$ for each $x \in \Omega$, even though from the observed distribution, we might expect $\zeta(Z=x|\Do(X=x)) = 1$.

Especially interesting are the extreme points of $I_{\pi}$.
In some senses, these are the interventional distributions which are maximally far away from the observational distribution, and thus are maximally `bad'.
All extreme points of $I_{\pi}$ correspond to models where $X$ and $U$ are \textbf{maximally correlated}, in the sense that for each $x \in \Omega$, there is a unique $u \in \Gamma$ so that $\mu_{x, u} > 0$.

\subsection{$\epsilon$-Perfect Channels and Diagonal Functionals}
We wish to weaken the condition that $X$ and $U$ be perfectly correlated to the case when rather than having $X$ and $Z$ be perfectly correlated, we instead have that there is some $\epsilon$ so that for each $x \in \Omega$,
\[
    \pi(Z=x|X=x)=\frac{\pi_{x,x}}{\pi_x} \ge 1-\epsilon
\]
We cannot describe all possible interventional distributions in this setting, but we can at least try to argue that it is not possible for the interventional distribution to be completely independent of $X$.
We say that an interventional distribution $\zeta$ is independent of $X$ if for each $x, x' \in \Omega$, and any $z \in \Omega$, we have that
\[
    \zeta(Z=z|X=x) = \zeta(Z=z|X=x').
\]

Recall that we assume that $|\Omega| = n$ and $|\Gamma| = k$, so that $X$ and $Z$ both take on $n$ values, but $U$ takes on at most $k$ values.

We can show the following:
\begin{lemma}\label{lem:eps_perf}
    If $X$ and $Z$ are $\epsilon$-perfectly correlated, and for each $x \in \Omega$, $\epsilon < 1-\frac{k}{n^2 \pi(X=x)}$, then for any $\zeta \in I_{\pi}$, $\zeta$ is not independent of $X$.
\end{lemma}

This is an immediate corollary of the next theorem, which provides an explicit hyperplane separator between $I_{\pi}$ and the set of interventional distributions independent of $X$.

We will consider linear functionals on the space of interventional distributions of the following form:
\[
    L(\zeta) = \sum_{x \in \Omega} \ell_x \zeta(Z=x|X=x)
\]
where $\ell_x \ge 0$ for all $x \in \Omega$.
We will call such a linear functional a diagonal functional.

If $L$ is such a diagonal functional, then let $\supp(L) = \{x \in \Omega : \ell_x > 0\}$.

\begin{reftheorem}(Theorem \ref{thm:diagonal_bound})
    Let $L$ be a diagonal functional, and let $\pi$ be some observed distribution.
    Let $\supp(L) = \{x \in \Omega : \ell_x > 0\}$, then for some $x \in \supp(L)$,
    \[
        L(\zeta) \ge \frac{|\supp(L)|^2}{k} \ell_x \pi(X=x, Z=x)
    \]
    for all $\zeta \in I_{\pi}$.
\end{reftheorem}

\begin{proof}(of lemma \ref{lem:eps_perf})
    Let $L$ be the linear functional so that $\ell_x = 1$ for each $x \in \Omega$, then under the assumptions of the corollary, we obtain that for any $\zeta \in L(\zeta)$, 
   \[
       L(\zeta) > 1
   \]
   On the other hand, if $\zeta$ is independent of $X$, then it can easily be seen that $L(\zeta) = 1$.
\end{proof}

\subsection{Bounded Mutual Information Between $X$ and $U$}
The previous results are concern situations when $k$ is much less than $n$, the distribution over $X$ is sufficiently close to being uniformly distributed, and $Z$ and $X$ are highly correlated.

We consider another result, which does not depend on $k$, nor on the relationship between $X$ and $Z$ but does depend on the mutual information between $X$ and $U$, a measure of the correlation between $X$ and $U$.
It may often be plausible to make this assumption in real world situations.

Given an observed distribution $\pi \in \Obs$, we will define $\eta \in \Int$ be the interventional distribution so that 
\[
    \eta(Z=z | X=x) = \pi(Z=z|X=x)
\]
That is, this is conditional distribution of $Z$ on $X$ in the observed distribution.

Given distributions $A \in \Delta(S_1)$ and $B \in \Delta(S_2)$, we also define their product to be $A \otimes B \in \Delta(S_1 \times S_2)$ where $(A \otimes B)_{i,j} = A_iB_j$.

If $\mu \in \Mu$ is some joint distribution on $\Omega \times \Gamma$ and we have that $\mu_X$ is the marginal distribution on $X$ and $\mu_U$ is the marginal on $U$, then we define the correlation distance of $\mu$ to be 
\[
    \delta(\mu) = \|\mu - \mu_X \otimes \mu_U\|_1
\]
where $\|\cdot\|_1$ is the $\mathcal{L}_1$ norm of a vector.

\begin{reftheorem}(Theorem \ref{thm:mutual_info})
    Suppose that $(\mu, \nu) \in M_{\pi}$ are model parameters with $\pi$ as their observed distribution. Let $\zeta$ be the associated interventional distribution.  Then,
    \[
        \|\zeta - \eta\|_1 \le \frac{1}{\min_{x} \pi(X=x)}\delta(\mu)
    \]

\end{reftheorem}

From this fact, it is an easy consequence of Pinsker's inequality (see \cite[Chapter 17]{cover1999elements}) that
\begin{lemma}
    Suppose that $(\mu, \nu) \in M_{\pi}$ are model parameters with $\pi$ as their observed distribution. Let $\zeta$ be the associated interventional distribution.  Then,
    \[
        \|\zeta - \eta\|_1 \le \frac{1}{\min_{x} \pi(X=x)}\sqrt{\frac{1}{2}I(X; U)}
    \]
\end{lemma}

This makes formal the intuitive idea that if $U$ and $X$ must be highly correlated if the causal distribution is significantly different from the observed distribution.

\section{Geometry of Interventional Distributions}
\label{sec:geometry}
A set is said to be  \textbf{simply-semialgebraic} if it can be defined as the set of points in $\R^n$ satisfying a finite set of polynomial inequalities.
A set is said to be \textbf{semialgebraic} if it is the union of finitely many simply-semialgebraic sets. The probability simplex for example is defined as 
\[
    \Delta(S) = \{(a_i)_{i \in S} : \forall i\;a_i \ge 0,\; \sum_{i\in S}a_i = 1\}
\]
Given an observed distribution $\pi \in \Obs$, we can obtain a semialgebraic description of the set of possible model parameters $(\mu, \nu)$ which are consistent with this observed distribution,
\[
    M_{\pi} = \{(\mu, \nu) \in \Mu \times \Nu : \forall z, x,\; \pi_{z,x} = \sum_{u} \mu_{u,x}\nu_{z,x,u}\}
\]

It is clear that this is also a simply-semialgebraic set.
In particular, it is defined entirely in terms of linear inequalities and quadratic equations.

We can define the interventional map, $\intmap: M_{\pi} \rightarrow I_{\pi}$ where $\intmap(\mu, \nu)$ is the distribution $\zeta$ so that
\[
\zeta_{x,z} = \sum_{u\in \Gamma} \mu_{u} \nu_{z,x,u}
\]
It is clear that $I_{\pi}$ is the image of this map.

It is a consequence of the Tarski-Seidenberg theorem (see \cite{bochnak2013real} for a reference) that $I_{\pi}$ is also semialgebraic. 
We will see in the next section on perfect channels that it is possible for $I_{\pi}$ to require exponentially many inequalities to define, and thus that we should not expect it to be tractable to obtain the full description of $I_{\pi}$ in terms of polynomial inequalities.

One observation of note is that if $\mu$ is \textbf{fixed}, then all of the above inequalities are linear in $\nu$.  This implies that if we fix a particular $\mu \in \Mu$, then the slice 
\[
    M_{\pi, \mu} = \{\nu \in \Nu : (\mu, \nu) \in M_{\pi}\}
\]
is a polytope, and the map $\intmap$ restricts to a linear map from $M_{\pi, \mu}$ to $I_{\pi}$.
The analogous facts also hold if we consider slices with fixed $\nu$.

The study of such parameterized polytopes is sometimes called geometric combinatorics, though we will not use any facts about this here \cite[Chapter 5]{blekherman2012semidefinite}.

\subsection{Linear Functionals}
We might attempt to find the maximum and minimum values of a linear functional $L$ over $I_{\pi}$.
Specifically, consider 
\[
    \min_{\zeta} \{L(\zeta) : \zeta \in I_{\pi, \mu}\} = \min_{\mu, \nu} \{L(\intmap(\mu, \nu)) : (\mu, \nu) \in M_{\pi, \mu}\}
\]
By expanding out this definition in terms of $\mu$ and $\nu$, we obtain the following optimization problem:
\begin{align*}
    \text{minimize }&\sum_{x \in \Omega} \sum_{z \in \Omega} \sum_{u \in \Gamma} \ell_{z, x} \left(\sum_{x' \in \Omega} \mu_{u,x'}\right) \nu_{z,x,u}\\
    \text{such that }& \forall x \in \Omega, \forall u \in \Gamma, \sum_{z \in \Omega} \nu_{z, x, u} = 1\\
                     &\forall x, z \in \Omega, \sum_{u \in \Gamma} \mu_{x, u} \nu_{z, x, u} = \pi_{x, z}\\
                     &\forall x,z \in \Omega, \forall u \in \Gamma, \nu_{z,x,u} \ge 0
\end{align*}
Because all of the constraints and objective are either linear or quadratic, this program is in fact a quadratically constrained quadratic program (\textbf{QCQP}).
By our previous observations, we note that if $\mu \in \Mu$ is fixed, and we regard this as a program purely in terms of $\nu$, we obtain a linear program.
A QCQP with this property is said to be bipartite bilinear \cite{dey2019new}, and it is possible to obtain both semidefinite programming and second order cone relaxations of this optimization problem.

For the purposes of our analysis, the linear program which is obtained by fixing $\mu \in \Mu$ will be important, so we define the $\mu-$\textbf{linear program} to be
\[
    P_{\mu} := 
     \min_{\nu} \{L(\intmap(\mu, \nu)) : (\mu, \nu) \in M_{\pi, \mu}\}
\]

In section \ref{sec:linear-program}, we will provide some results about the nature of this $\mu$-linear program, and in particular compute its dual program.

\section{Perfect Channels}
In this section, we will consider an observational distribution which has no noise in the following sense.

Let $\pi \in \Obs$ have the property that for all $x\neq z \in \Omega$, $\pi_{x,z} = 0$.
Such an observational distribution will be called a perfect channel.

In this case, we will seek a complete characterization of $M_{\pi}$ and $I_{\pi}$ as the finite union of polytopes in \textbf{theorem \ref{thm:perfect_full}}.
We will also obtain a slightly more compact representation of the extreme points of $I_{\pi}$ in theorem \textbf{theorem \ref{thm:perfect_conv}}.
The polytopes in this decomposition will correspond to different patterns of zeros in the distribution $\mu$.

For any $\mu \in \Mu$, let the support of $\mu$ be 
\[
    \sigma(\mu) = \{(x,u) \in \Omega \times \Gamma : \mu_{x,u} > 0\}.
\]
Now, fix some set $S \subseteq \Omega \times \Gamma$, so that for each $x \in \Omega$, there is some $u \in \Gamma$ with $(x,u) \in S$.
Then let 
\[
    \Sigma_S = \{\mu \in \Mu : \forall x\in \Omega, \mu_x = \pi_x, \text{ and }\sigma(\mu) \subseteq S\}, \text{ and}
\]
\[
    N_{S} = \{\nu \in \Nu : \forall (x, u) \in S, \forall z \in \Omega,\;\nu_{z,x,u} = 1_{z = x}\}.
\]
Observe that any face of $\Delta(\Omega \times \Gamma)$ is of the form $\Sigma_S$ for some $S$.

Let $Q_S = \{\intmap(\mu, \nu) : \mu \in \Sigma_S, \nu \in N_S\}$.
Our main result in this section is that 
\begin{theorem}\label{thm:perfect_full}
    \[
        I_{\pi} = \bigcup_{S \subseteq \Omega \times \Gamma} Q_S, 
    \]
    and $Q_S$ is a convex polytope.
\end{theorem}

To prove this theorem, we state a few lemmas.

\begin{lemma}\label{lmma:support}
    If $\pi$ is a perfect channel, then $(\mu, \nu) \in M_{\pi}$ if and only if for some $S \subseteq \Omega \times \Gamma$, $\mu \in \Sigma_S$ and $\nu \in N_S$.
    In particular, $I_{\pi} = \bigcup_{S \subseteq \Omega \times \Gamma} Q_S$.
\end{lemma}
\begin{proof}
    Interpreted probabilistically, this states the following: if the pair $(x,u)$ appears with nonzero probability in the observed distribution, then because $Z = X$ with probability 1 in the observed distribution, we must have that $\nu_{x,x,u} = 1$ in the observed distribution as well.

    Formally, for $x \neq z$, note that
    \[
        \pi_{x, z} = \sum_{u' \in \Gamma} \mu_{x, u'} \nu_{z,x,u'} = 0
    \]

    Because $\mu, \nu \ge 0$, we have that $\pi_{x,z} = 0$ iff for each $u \in \Gamma$, either $\mu_{x,u} = 0$ or $\nu_{z,x,u} = 0$.
    This implies that if $\mu_{x,u} > 0$, then $\nu_{z,x,u} = 0$.

    Formally, for $x = z$, note that
    \[
        \pi_{x, z} = \sum_{u' \in \Gamma} \mu_{x, u'} \nu_{z,x,u'} = 1
    \]
    Again, because we have that $0\le \nu_{z,x,u'} \le 1$, we have that if $\mu_{x,u'} > 0$, then $\nu_{z,x,u'} =1$.

    This, combined with the equation of marginals $\mu_x = \pi_x$, is equivalent to the condition that $(\mu, \nu) \in M_{\pi}$.

    Therefore, 
    \[
        M_{\pi} = \bigcup_{S \subseteq \Omega \times \Gamma} (\Sigma_S \times N_S),
    \]
    It is then clear that 
    \[
        I_{\pi} = \bigcup_{S \subseteq \Omega \times \Gamma} \intmap(\Sigma_S, N_S)
    \]
\end{proof}

We now show that $Q_S$ is a convex polytope.
\begin{lemma}
    \label{lmma:poly}
    $Q_S$ is convex.
\end{lemma}
\begin{proof}
    Suppose that $\zeta_1 = \intmap(\mu_1,\nu_1)$ and $\zeta_2=\intmap(\mu_2,\nu_2)$, and let $\zeta' = t \zeta_1 + (1-t)\zeta_2$. We want to show that there are $\mu' \in \Sigma_S$ and $\nu' \in N_S$ so that $\zeta' = \intmap(\mu', \nu')$.

    We will let $\mu' = t \mu_1 + (1-t)\mu_2$, and because $\Sigma_S$ is convex, $\mu' \in \Sigma_S$. We define $\nu'$ as follows: 

    If $(x,u) \in S$, then for each $z \in \Omega$, let $\nu'_{z,x,u} = 1_{z=x}$. As long as this is the case, and $\nu'$ defines a valid element of $\Nu$, then $\nu'$ also defines an element of $N_S$.

    Defining $\nu'_{z,x,u}$ if $(x,u) \not \in S$ is a little more intricate. Fix $(x,u) \not \in S$. Let $1_{z=x}$ be the indicator function for $z = x$, and consider
    \[
        \nu'_{z,x,u} = \frac{\zeta'_{x,z} - \sum_{u' : (x,u') \in S}\mu_u 1_{z=x}}{\sum_{u' : (x,u') \not \in S} \mu_u}
    \]

    We need to check that $\nu'_{z,x,u} \ge 0$ for each $z,x,u$ and that $\sum_{z \in \Omega}\nu'_{z,x,u} = 1$.

    To check nonnegativity, note that
    \[
        (\zeta_1)_{z,x} = \sum_{u' \in \Gamma} (\mu_1)_{u'} (\nu_1)_{z,x,u'}
    \]
    \[
        (\zeta_2)_{z,x} = \sum_{u' \in \Gamma} (\mu_2)_{u'} (\nu_2)_{z,x,u'}
    \]
    Thus, because $\nu_1, \nu_2 \in N_S$,
    \[
        (\zeta_1)_{z,x} - \sum_{u' : (x,u') \in S} (\mu_1)_{u'} 1_{z=x} = \sum_{u' : (x,u') \not \in S} (\mu_1)_{u'} (\nu_1)_{z,x,u'} \ge 0
    \]
    \[
        (\zeta_2)_{z,x} - \sum_{u' : (x,u') \in S} (\mu_2)_{u'}1_{z=x} = \sum_{u' : (x,u') \not \in S} (\mu_2)_{u'} (\nu_2)_{z,x,u'} \ge 0
    \]

    So,
    \[
        \zeta'_{x,z} - \sum_{u : (x,u) \in S}\mu_u 1_{z=x} = t((\zeta_1)_{x,z} -  \sum_{u : (x,u) \in S}(\mu_1)_u 1_{z=x}) + (1-t)((\zeta_1)_{x,z} -  \sum_{u : (x,u) \in S}(\mu_1)_u 1_{z=x}) \ge 0
    \]

    To check that this sums to 1, note that
    \begin{align*}
        \sum_{z \in \Omega} \nu'_{z,x,u} &= \sum_{z\in \Omega}\left(\frac{\zeta'_{x,z} - \sum_{u : (x,u) \in S}\mu_u 1_{z=x}}{\sum_{u : (x,u) \not \in S} \mu_u}\right)\\
                             &= \frac{\sum_{z\in \Omega}\left(\zeta'_{x,z} - \sum_{u : (x,u) \in S}\mu_u 1_{z=x}\right)}{\sum_{u : (x,u) \not \in S} \mu_u}\\
                             &= \frac{1 - \sum_{u : (x,u) \in S}\mu_u}{\sum_{u : (x,u) \not \in S} \mu_u}\\
                             &= 1
    \end{align*}
    As desired.
\end{proof}

\begin{lemma}\label{lmma:extreme_pts}
    Every extreme point of $Q_S$ is of the form $\intmap(\mu, \nu)$, where $\mu$ and $\nu$ are extreme points of $\Sigma_S$ and $N_S$, respectively.
\end{lemma}
\begin{proof}
    Let $\zeta = \intmap(\mu, \nu)$.

    Because $\mu \in \Sigma_S$, we can write $\mu$ as a convex combination of extreme points of $\Sigma_S$.
    \[
        \mu = \sum_{i=1}^k t_i\mu^i
    \]
    where $\sum_{i=1}^k t_i = 1$ and $t_i \ge 0$ for each $i$, and each $\mu^i$ is an extreme point of $\Sigma_S$.
    
    Similarly, 
    \[
        \nu = \sum_{i=1}^{\ell} s_i\nu_i
    \]
    where $s_i \ge 0$, $\sum_{i=1}^k s_i$, and $\nu_i$ are extreme points of $N_S$.

    From this, and the bilinearity of $\intmap$, we obtain that
    \[
        \zeta = \intmap(\mu,\nu) = \sum_{i=1}^k\sum_{j=1}^{\ell} t_is_j \intmap(\mu_i,\nu_j)
    \]

    Thus, $\zeta$ can be written as a convex combination of points of the form $\intmap(\nu, \mu)$, where $\nu$ and $\mu$ are both extreme points. Thus, the only extreme points of $Q_S$ are those of the desired form.
\end{proof}

\begin{proof}(of theorem \ref{thm:perfect_full})
    $\Sigma_S$ and $N_S$ are both polytopes, since they are defined by finitely many linear inequalities, so they each have finitely many extreme points, and thus, by lemmas  \ref{lmma:poly} and \ref{lmma:extreme_pts}, $Q_S$ is a convex set with finitely many extreme points and thus is a convex polytope.

    The theorem follows from lemma \ref{lmma:support}.
\end{proof}

\subsection{Extreme Points of $I_{\pi}$}

We can be somewhat more explicit about what the extreme points of $\Sigma_S$ and $N_S$ are.

A function $f : \Omega \rightarrow \Gamma$ is \textbf{contained in $S \subseteq \Omega \times \Gamma$} if $(x,f(x)) \in S$ for each $x \in \Omega$. We define $\mu^f$ to be the point
\[
    (\mu^f)_{x, u} = \pi_x 1_{u = f(x)}
\]
Analogously, if $f : \Omega \times \Gamma \rightarrow \Omega$ is a function, then we say that $f$ is \textbf{compatible with $S$} if for all $(x,u) \in \Gamma$, $f(x,u) = x$, and we let
\[
    (\nu^f)_{z, x, u} = 1_{z = f(x,u)}
\]

\begin{lemma}
    Every extreme point of $\Sigma_S$ is of the form $\mu^f$ where $f$ is a function contained in $S$.
\end{lemma}

\begin{lemma}
    Every extreme point of $N_S$ is of the form $\nu^f$ where $f$ is a function compatible with $S$.
\end{lemma}

In particular, we see that all of the extreme points of $I_{\pi}$ are of the form $\intmap(\mu_f, \nu_g)$, where there is a set $S$ so that $f$ is contained in $S$ and $g$ is compatible with $S$.

Let $S_f = \{(x,f(x)) : x \in \Omega\}$, and notice that if $f$ is contained in $S$, and $g$ is compatible with $S$, then $g$ is compatible with $S_f \subseteq S$. Therefore, there is a set $S$ so that $f$ is contained in $S$ and $g$ is compatible with $S$ if and only if $g$ is compatible with $S_f$.

In particular, for a function $f \in \Gamma^{\Omega}$, let
\[
    Q_f = \intmap(\Sigma_{S_f}, N_{S_f})
\]
$\conv(I_{\pi}) \subseteq \conv(\bigcup_{f \in \Gamma^\Omega} Q_f)$, where this union is over all functions from $\Omega$ to $\Gamma$.

Finally, we can characterize the inequalities defining $\intmap(\Sigma_{S_f}, N_{S_f})$.
\begin{theorem}\label{thm:perfect_conv}
    If $\zeta \in \Int$, then $\zeta \in \intmap(\Sigma_{S_f}, N_{S_f})$ if and only if 
    \[
        \zeta_{x,x} \ge \sum_{x' \in \Omega : f(x) = f(x')} \pi_x
    \]
\end{theorem}
\begin{proof}
    Notice that if $\mu \in \Sigma_{S_f}$ if and only if $\mu = \mu^f$, and that for each $x \in \Omega$,
    \[
        \mu^f_{f(x)} = \sum_{x' \in \Omega : f(x) = f(x')} \pi_x
    \]

    To show one direction, let $\zeta = \intmap(\mu, \nu)$, where $\mu \in \Sigma_{S_f}$, and $\nu \in N_{S_f}$, then
    \begin{align*} 
        \zeta_{x,x}&=\sum_{u \in \Gamma} \mu_u \nu_{x,x,u}\\
                   &\ge\mu_{f(x)} \nu_{x,x,f(x)}\\
                   &= \sum_{x' \in \Omega : f(x) = f(x')} \pi_{x'}
    \end{align*}

    It remains to show that if $\zeta$ satisfies these inequalities, then it is in $\intmap(\Sigma_{S_f}, N_{S_f})$. 
    Let 
    \[
        \nu_{z,x,u} =
        \begin{cases}
            1_{z=x} \text{ if }u=f(x)\\
            \frac{\zeta_{z,x} - \mu^f_{f(x)}1_{z=x}}{1 - \mu^f_{f(x)}} \text{ otherwise }
        \end{cases}
    \]
    We want to show that this value of $\nu$ is in $N_{S_f}$, and $\intmap(\mu^f, \nu) = \zeta$. 

    To see that this is in $\Nu$, notice that it is nonnegative for each entry, and it is a simple computation to see that
    \[
        \sum_{z \in \Omega} \nu_{z,x,u} = 1
    \]

    That it is in $N_{S_f}$ also follows easily from the definition.

    On the other hand, let $\zeta' = \intmap(\mu^f, \nu)$, so that
    \begin{align*} 
    \zeta'_{z,x}&=\sum_{u} \mu_u \nu_{z,x,u}\\
            &=\mu_{f(x)}1_{z=x} + \sum_{u \neq f(x)} \mu_u \left(\frac{\zeta_{z,x} - \mu^f_{f(x)}1_{z=x}}{1 - \mu^f_{f(x)}}  \right)\\
            &=\mu_{f(x)}1_{z=x} + \left(\sum_{u \neq f(x)} \mu_u\right) \left(\frac{\zeta_{z,x} - \mu^f_{f(x)}1_{z=x}}{1 - \mu^f_{f(x)}}  \right)\\
            &= \zeta_{z,x}
    \end{align*}
    This gives the desired result.
\end{proof}

%
\section{Analysis of the $\mu$-Linear Program}\label{sec:linear-program}
Recall from above that for any $\mu \in \Mu$, we have an associated linear program, $P_{\mu}$, defined as
\begin{align*}
    \text{minimize }&\sum_{x \in \Omega} \sum_{z \in \Omega} \sum_{u \in \Gamma} \ell_{z, x} \mu_u \nu_{z,x,u}\\
    \text{such that }& \forall x \in \Omega, \forall u \in \Gamma, \sum_{z \in \Omega} \nu_{z, x, u} = 1\\
                     &\forall x, z \in \Omega, \sum_{u \in \Gamma} \mu_{x, u} \nu_{z, x, u} = \pi_{x, z}\\
                     &\forall x,z \in \Omega, \forall u \in \Gamma, \nu_{z,x,u} \ge 0
\end{align*}
\begin{lemma}
    For each $\mu \in \Mu$, $P_{\mu}$ is bounded and feasible.
\end{lemma}\label{lmma:bound_feas}
\begin{proof}
    Observe that the feasible set for $P_{\mu}$ is a closed subset of $\Mu$, which implies that the feasible region is compact, and hence, that $P_{\mu}$ is bounded.

    On the other hand, let $\nu \in \Nu$ be such that
    \[
        \nu_{z,x,u} = \frac{\pi_{x, z}}{\pi_x}.
    \]
    That is, $Z$ is independent of $U$ and simply chosen from the true conditional distribution of $Z$ given $X$.

    From this, it is clear that $\nu$ is in fact compatible with $\mu$ for any $\mu \in  \Mu$, and hence feasible point of $P_{\mu}$.
\end{proof}

From the previous lemma, we can obtain \textbf{strong duality} \cite[Chapter 7]{schrijver1998theory}, which states that the value of $P_{\mu}$ is equal to that of the following program: 
\begin{align*}
    P_{\mu}^* = 
    \text{maximize }&\sum_{x \in \Omega} \sum_{z \in \Omega} \pi_{x, z} \alpha_{x, z} + \sum_{x \in \Omega} \sum_{u \in \Gamma} \beta_{u, x}\\
    \text{such that }&\forall x,z \in \Omega, \forall u \in \Gamma,
    \mu_{x,u} \alpha_{x, z} + \beta_{u,x} \le \ell_{z, x}\mu_u
\end{align*}
We can greatly simplify this dual program, to turn it into an unconstrained optimization problem.

Let 
\[
    f(\mu, \alpha) = \sum_{x \in \Omega} \sum_{z\in \Omega} \pi_{x,z}\alpha_{x,z} + \sum_{x \in \Omega} \sum_{u \in \Gamma} \min_{z \in \Omega} \left( \ell_{z, x}\mu_u - \mu_{x,u} \alpha_{x,z}\right)
\]

\begin{lemma}\label{lmma:dual_bound}
    \[\val(P_{\mu}) = \max_{\alpha \in \R^{\Omega^2}} f(\mu, \alpha)\]
\end{lemma}
\begin{proof}
    Fix $\alpha \in \R^{\Omega^2}$, and let $\beta$ be defined so that
    \[
        \beta_{u,x} = \min_{z \in \Omega}\left( \ell_{z,x} \mu_u - \mu_{x,u} \alpha_{x, z}\right)
    \]
    It is not hard to see that the pair $(\alpha, \beta)$ are a feasible point of $P_{\mu}^*$, and so the value of $P_{\mu}^*$ is at least $f(\mu, \alpha)$.

    On the other hand, given an optimum $(\alpha, \beta)$ of $P_{\mu}$, notice that the inequalities imply that
    \[
        \beta_{u,x} \le \min_{z \in \Omega}\left( \ell_{z,x} \mu_u - \mu_{x,u} \alpha_{x, z}\right).
    \]
    If this inequality were strict, then we could increase $\beta_{u,x}$ without violating the inequalities, increasing the objective, so at an optimum, we must have that in fact, that this is inequality is an equality.
    We see that at this point, the objective is then equal to $f(\mu, \alpha)$, giving that $\val(P_{\mu}^*) \le \max_{\alpha} f(\mu, \alpha)$.
\end{proof}

\subsection{Mutual Information Bound}
Recall that if $\mu$ is a distribution in $\Mu$, so that the marginal on $X$ is $\mu_X$ and the marginal on $U$ is $\mu_U$, then the correlation distance of $\mu$ is 
\[
    \delta(\mu) = \|\mu - \mu_X \otimes \mu_U\|_1
\]
where $\otimes$ denotes the product distribution.

We recall that the mutual information between two random variables $U$ and $X$ distributed according to $\mu$ is defined in terms of the divergence between $\mu$ and the product distribution $\mu_X \otimes \mu_U$ (see \cite{cover1999elements}):
\[
    I(U ; X) = D(\mu || \mu_X \otimes \mu_U).
\]

We also recall the Pinsker inequality, which implies that
\[
    \delta(\mu) \le \sqrt{\frac{1}{2} I(U;X)}
\]

Also recall that the honest distribution $\eta \in \Int$ is defined so that
\[
    \eta_{x,z} = \frac{\pi_{x,z}}{\pi_x}
\]

using these results and the results in the previous section, we can show the following bound on $\val(P_{\mu})$:

\begin{theorem}
    \label{thm:mutual_info}
    Suppose that $(\mu, \nu) \in M_{\pi}$ are model parameters with $\pi$ as their observed distribution. Let $\zeta$ be the associated interventional distribution.  Then,
    \[
        \|\zeta - \eta\|_1 \le \frac{1}{\min_{x} \pi(X=x)}\delta(\mu)
    \]
\end{theorem}
\begin{proof}
    Let $\ell_{x,z} = -\sgn((\zeta - \eta)_{x,z})$, where $\sgn(x) =\begin{cases} 1\text{ if }x \ge 0\\-1\text{ otherwise}\end{cases}$.

    Let $L$ be the linear functional so that for any $\psi \in \Int$, $L(\psi) = \langle \ell, \psi\rangle$, and notice that 
    \[
        L(\zeta - \eta) = L(\zeta) - L(\eta) = -\|\zeta - \eta\|_1
    \]
    In particular, because the right side of this equation is negative, we can assume that $L(\zeta) \le L(\eta)$.

    Hence, it suffices to obtain an lower bound on $L(\zeta) - L(\eta)$ for any $\zeta$ which is compatible with this value of $\mu$. 

    Let $P_{\mu}$ be the linear program with $L$ as its objective. It suffices to lower bound $\val(P_{\mu})$.

    Using the results in the previous section, we have that for any $\alpha \in \R^{\Omega^2}$,
    \[
        \val(P_{\mu}) \ge f(\mu, \alpha).
    \]
    We now choose the dual variables  $\alpha_{x,z} = \frac{\ell_{x,z}}{\pi_x}$.
    At this $\alpha$, $f$ evaluates to
    \[
        f(\mu, \alpha) = \sum_{x \in \Omega} \sum_{z\in \Omega} \frac{\pi_{x,z}}{\pi_x}\ell_{x,z} + \sum_{x \in \Omega} \sum_{u \in \Gamma} \min_{z \in \Omega} \left( \ell_{z, x}\left(\mu_u - \frac{\mu_{x,u}}{\pi_x}\right)\right)
    \]
    Let $d_{x,u} = \mu_{x,u} - \mu_u \pi_x$. We have that
    \[
        \sum_{x \in \Omega}\sum_{u\in \Gamma} |d_{x,u}| = \delta(\mu)
    \]

    We compute
    \begin{align*}
        |L(\eta) - f(\mu, \alpha)| &= \left|\sum_{x \in \Omega} \sum_{u \in \Gamma} \min_{z \in \Omega} \left( \ell_{z, x}\left(\mu_u - \frac{\mu_u\pi_{x} + d_{x,u}}{\pi_x}\right)\right)\right|\\
                                   &\le \sum_{x \in \Omega} \sum_{u \in \Gamma} \left|\min_{z \in \Omega} \left( -\frac{\ell_{z, x}d_{x,u}}{\pi_x}\right)\right|\\
                                   &\le \sum_{x \in \Omega} \sum_{u \in \Gamma} \frac{1}{\min_x \pi_x} |d_{x,u}|\\
                                   &\le \frac{1}{\min_x \pi_x} \delta(\mu)
    \end{align*}

    Hence, we have that $L(\zeta) - L(\eta) \ge -\frac{1}{\min_x \pi_x} \delta(\mu)$, implying that
    \[
        \|\zeta - \eta\|_1 \le \frac{1}{\min_x \pi_x} \delta(\mu)
    \]

\end{proof}
\section{Diagonal Linear Functionals}
We will apply the results of the last section on $P_{\mu}$ to a particular kind of linear functional on the space $\Int$.

Let $L$ be a linear functional of the form
\[
    L(\zeta) = \sum_{x \in \Omega} \ell_x \zeta_{x,x}.
\]

We will also assume that $|\Gamma| \le k$, so that $U$ can take on at most $k$ distinct values.

We will make use of an easy lemma:
\begin{lemma}
    \label{lmma:amhm}
    If $a_1 ,\dots, a_{\ell} \ge 0$, then 
    \[
        \sum_{i=1}^{\ell} \frac{1}{a_i} \ge \frac{\ell^2}{\sum_{i=1}^{\ell} a_i} 
    \]

\end{lemma}
\begin{proof}

    \begin{align*} 
        \sum_{i=1}^{\ell}\frac{1}{a_i} &= \ell \sum_{i=1}^{\ell}\frac{1}{\ell a_i}\\
                                       &\ge \frac{\ell}{\sum_{i=1}^{\ell}\frac{1}{\ell}a_i}\\
                                       &=  \frac{\ell^2}{\sum_{i=1}^{\ell}a_i},\\
    \end{align*}
    where the second line follows from the arithmethic-mean-harmonic-mean (AMHM) inequality \cite{hardy1952inequalities}.
    
\end{proof}
    We focus on the sum of the form
    \[
          \sum_{x \in g^{-1}(u)}\frac{1}{\mu_{g(x),x}}.
    \]
    We wish to apply the arithmetic-mean-harmonic-mean inequality here , so we perform some manipulations:

\begin{theorem}
    \label{thm:diagonal_bound}
    Let $L$ be a diagonal functional, and let $\pi$ be some observed distribution.
    Let $\supp(L) = \{x \in \Omega : \ell_x > 0\}$, then for some $x \in \supp(L)$,
    \[
        L(\zeta) \ge \frac{|\supp(L)|^2}{k} \ell_x \pi(X=x, Z=x)
    \]
    for all $\zeta \in I_{\pi}$.
\end{theorem}
\begin{proof}
    From lemma \ref{lmma:dual_bound}, we have that for any $\mu \in \Mu$, and any $\alpha \in \R^{\Omega^2}$,
    \[
        \val(P_{\mu}) \ge f(\mu, \alpha)
    \]

    Let $g : \supp(L) \rightarrow \Gamma$ be the function $g(x) = \argmin\frac{\mu_u}{\mu_{u,x}}$, and then let
    \[
        \alpha_{x,z} = 
        \begin{cases}
            \frac{\ell_x\mu_{g(x)}}{\mu_{g(x),x}} \text{ if } x = z\\
            0 \text{ otherwise}
        \end{cases}
    \]
    Inputting this particular value of $\alpha$ and $L$ into $f$ yields
    \[
        f(\mu, \alpha) = \sum_{x \in \Omega} \pi_{x,x}\alpha_{x,x} + \sum_{x \in \Omega} \sum_{u \in \Gamma} \min\{ \ell_{x}\mu_u - \mu_{x,u} \alpha_{x,x}, 0\}
    \]
    From our choice of $\alpha$, we have that for any $x, u$, 
    \[
        \ell_{x}\mu_u - \mu_{x,u} \alpha_{x,x} \ge 0
    \]
    Hence, we have that 
    \begin{align*} 
        f(\mu, \alpha)&\ge\sum_{x \in \supp(L)} \frac{\ell_x\pi_{x,x}\mu_{g(x)}}{\mu_{g(x),x}}\\
                      &=\sum_{u \in \Gamma} \sum_{x \in g^{-1}(u)}\frac{\ell_x\pi_{x,x}\mu_{g(x)}}{\mu_{g(x),x}}\\
                      &\ge \left(\min_{x \in \supp(L)} \ell_x \pi_{x,x}\right)\left(\sum_{u \in \Gamma} \mu_{u}\sum_{x \in g^{-1}(u)}\frac{1}{\mu_{u,x}}\right).
    \end{align*}

    By lemma \ref{lmma:amhm}, for any fixed $u \in \Gamma$,
    \[
        \sum_{x \in g^{-1}(u)}\frac{1}{\mu_{u,x}} \ge \frac{|g^{-1}(u)|^2}{\sum_{x \in g^{-1}(u)}\mu_{u, x}} \ge \frac{|g^{-1}(u)|^2}{\mu_{u}}.
    \]

    Combining the previous two inequalities yields
    \begin{align*} 
        f(\mu, \alpha) &\ge \left(\min_{x \in \supp(L)} \ell_x \pi_{x,x}\right)\left(\sum_{u \in \Gamma} |g^{-1}(u)|\right).
    \end{align*}

    Notice that $\sum_{u \in \Gamma} |g^{-1}(u)| = n$, and that the function $\|x\|^2 = \sum_{u \in \Gamma} x_u^2$ is convex.
    Therefore, we can obtain the bound
    \[
        \sum_{u \in \Gamma} |g^{-1}(u)|^2 \ge \frac{|\supp(L)|^2}{k}.
    \]

    The conclusion then follows.

\end{proof}

\section{Acknowledgements}
We would like to thank Leonard Schulman for initiating this line of inquiry and for valuable discussions. We would like to thank Spencer Gordon for helpful discussions.

\bibliographystyle{plain}
\bibliography{causal}

\begin{thebibliography}{10}

\bibitem{blekherman2012semidefinite}
Grigoriy Blekherman, Pablo~A Parrilo, and Rekha~R Thomas.
\newblock {\em Semidefinite optimization and convex algebraic geometry}.
\newblock SIAM, 2012.

\bibitem{bochnak2013real}
Jacek Bochnak, Michel Coste, and Marie-Fran{\c{c}}oise Roy.
\newblock {\em Real algebraic geometry}, volume~36.
\newblock Springer Science \& Business Media, 2013.

\bibitem{cover1999elements}
Thomas~M Cover.
\newblock {\em Elements of information theory}.
\newblock John Wiley \& Sons, 1999.

\bibitem{dey2019new}
Santanu~S Dey, Asteroide Santana, and Yang Wang.
\newblock New socp relaxation and branching rule for bipartite bilinear
  programs.
\newblock {\em Optimization and Engineering}, 20(2):307--336, 2019.

\bibitem{garcia2005algebraic}
Luis~David Garcia, Michael Stillman, and Bernd Sturmfels.
\newblock Algebraic geometry of bayesian networks.
\newblock {\em Journal of Symbolic Computation}, 39(3-4):331--355, 2005.

\bibitem{hardy1952inequalities}
G.H. Hardy, Karreman Mathematics~Research Collection, J.E. Littlewood,
  G.~P{\'o}lya, G.~P{\'o}lya, and D.E. Littlewood.
\newblock {\em Inequalities}.
\newblock Cambridge Mathematical Library. Cambridge University Press, 1952.

\bibitem{huang2012pearl}
Yimin Huang and Marco Valtorta.
\newblock Pearl's calculus of intervention is complete.
\newblock {\em arXiv preprint arXiv:1206.6831}, 2012.

\bibitem{LeeSpekkens}
Ciarán~M. Lee and Robert~W. Spekkens.
\newblock Causal inference via algebraic geometry: Feasibility tests for
  functional causal structures with two binary observed variables.
\newblock {\em Journal of Causal Inference}, 5(2), 2017.

\bibitem{pearl2009causality}
Judea Pearl.
\newblock {\em Causality}.
\newblock Cambridge university press, 2009.

\bibitem{schrijver1998theory}
Alexander Schrijver.
\newblock {\em Theory of linear and integer programming}.
\newblock John Wiley \& Sons, 1998.

\bibitem{schulman2016stability}
Leonard~J Schulman and Piyush Srivastava.
\newblock Stability of causal inference.
\newblock 2016.

\bibitem{shpitser2006identification}
Ilya Shpitser and Judea Pearl.
\newblock Identification of joint interventional distributions in recursive
  semi-markovian causal models.
\newblock In {\em Proceedings of the National Conference on Artificial
  Intelligence}, volume~21, page 1219. Menlo Park, CA; Cambridge, MA; London;
  AAAI Press; MIT Press; 1999, 2006.

\end{thebibliography}

\end{document}